\documentclass[conference]{IEEEtran}
\IEEEoverridecommandlockouts

\usepackage{cite}
\usepackage{amsmath,amssymb,amsfonts}
\usepackage{textcomp}
\usepackage{xcolor}
\def\BibTeX{{\rm B\kern-.05em{\sc i\kern-.025em b}\kern-.08em
    T\kern-.1667em\lower.7ex\hbox{E}\kern-.125emX}}

\usepackage{xspace}

\usepackage{pifont}

\usepackage{tikz}

\usepackage{amsmath,amssymb,amsfonts}

\usepackage{booktabs}

\usepackage{amsthm}

\usepackage{algorithm}
\usepackage[noend]{algpseudocode}
\algrenewcommand\textproc{}
\algrenewcommand\alglinenumber[1]{\tiny #1:}

\usepackage[caption=false,font=normalsize,labelfont=sf,textfont=sf]{subfig}

\usepackage{orcidlink}

\makeatletter
\g@addto@macro{\UrlBreaks}{\UrlOrds}
\makeatother

\hypersetup{
    colorlinks=true,
    linkcolor=black,
    citecolor=black,
    urlcolor=black,
    breaklinks=true,
}

\newcommand{\sysname}{Harpagon\xspace}
\newcommand{\sysabbr}{Harp\xspace}

\newcommand{\evalnum}{1131\xspace}

\newcommand\paraspace{\vspace*{0.25ex}}
\providecommand\parab[1]{\paraspace\noindent\textbf{#1}}
\providecommand\parae[1]{\paraspace\textbf{\textit{#1}}}

\newcommand{\ie}{\emph{i.e.,}\xspace}
\newcommand{\eg}{\emph{e.g.,}\xspace}

\newcommand{\secref}[1]{\S\ref{#1}}
\newcommand{\figref}[1]{Fig.~\ref{#1}}
\newcommand{\tabref}[1]{Table~\ref{#1}}
\newcommand{\algoref}[1]{Algorithm~\ref{#1}}

\newcommand{\lineref}[1]{Line~\ref{#1}}
\newcommand{\multilineref}[2]{Line~\ref{#1}-\ref{#2}}

\newcommand{\theoremref}[1]{Theorem~\ref{#1}}

\newcommand{\mynum}[1]{\textcolor[RGB]{0,0,0}{\ding{#1}}} 

\newcommand{\codesub}[2]{\texttt{#1}\textsubscript{\texttt{#2}}}

\theoremstyle{definition}

\newtheorem{theorem}{Theorem}

\newcommand{\tight}{\vspace{-3pt}}
\begin{document}

\title{
\sysname: Minimizing DNN Serving Cost via Efficient Dispatching, Scheduling and Splitting
\thanks{The work of Yitao Hu was supported by National Key Research and Development Program of China (2022YFB4501000) and National Natural Science Foundation of China under Grant Nos.62202328.}%
\thanks{Corresponding author: Yitao Hu (email: \href{mailto:yitao@tju.edu.cn}{yitao@tju.edu.cn}).}
}



\author{
    \IEEEauthorblockN{
        Zhixin Zhao\textsuperscript{*},
        Yitao Hu\textsuperscript{*},
        Ziqi Gong\textsuperscript{*},
        Guotao Yang\textsuperscript{*},
        Wenxin Li\textsuperscript{*},
        Xiulong Liu\textsuperscript{*},
        Keqiu Li\textsuperscript{*},
        Hao Wang\textsuperscript{\dag},
    }
    \IEEEauthorblockA{
        \textsuperscript{*}Tianjin Key Laboratory of Advanced Networking, Tianjin University, China \\
        \textsuperscript{\dag}Stevens Institute of Technology, USA
    }
}

\maketitle

\begin{abstract}

Advances in deep neural networks (DNNs) have significantly contributed to the development of real-time video processing applications. Efficient scheduling of DNN workloads in cloud-hosted inference systems is crucial to minimizing serving costs while meeting application latency constraints. However, existing systems suffer from excessive module latency during request dispatching, low execution throughput during module scheduling, and wasted latency budget during latency splitting for multi-DNN application, which undermines their capability to minimize the serving cost.

In this paper, we design a DNN inference system called \sysname, which minimizes the serving cost under latency constraints with a three-level design. It first maximizes the batch collection rate with a batch-aware request dispatch policy to minimize the module latency. It then maximizes the module throughput with multi-tuple configurations and proper amount of dummy requests. It also carefully splits the end-to-end latency into per-module latency budget to minimize the total serving cost for multi-DNN applications. Evaluation shows that \sysname outperforms the state of the art by $1.49$ to $2.37$ times in serving cost while satisfying the latency objectives. Additionally, compared to the optimal solution using brute force search, \sysname derives the lower bound of serving cost for $91.5\%$ workloads with millisecond level runtime.

\end{abstract}

\section{Introduction}
\label{sec:intro}
Deep neural network (DNN) has achieved the state of the art results in multiple fields~\cite{lecun2015deep}, such as computer vision, speech recognition and natural language processing. Therefore, the application developers are starting to build streaming applications using DNN models~\cite{shen2019nexus, hu2021scrooge, crankshaw2020inferline, crankshaw2017clipper, romero2021infaas}. These DNN-based applications usually have latency constraints to guarantee good user experiences~\cite{gunasekaran2022cocktail, wu2023graft, chen2024harmonybatch, peng2024tangram, dou2024emma}. For example, the traffic monitoring application~\cite{zhang2017live}, which extracts pedestrian and vehicle information from videos collected by the surveillance cameras, often expects results to be returned within a few milliseconds to provide timely supports when emergency occurs.

To satisfy the latency objective, the DNN-based applications are usually deployed in the inference system to leverage the plentiful GPU resources in the cloud~\cite{weng2022mlaas, zhang2022workload, wu2023transparent, wang2023proactive, xu2022igniter}. The inference system uses the classic \textit{client-server} architecture, where the client sends requests to the inference system to query the target DNN models with latency constraints. When providing services to the clients, the application developer is charged by the cloud provider according to the amount of computation resources used by the DNN-based application. Inference for DNN-based application accounts for nearly $90\%$ of computing costs in the cloud~\cite{aws_ec2}. Therefore, when running the inference system, the primary goal for the application developer is to find a proper DNN module configuration (\eg the batch size for DNN inference), so as to minimize the serving cost while satisfying the latency objective.


However, existing serving systems~\cite{shen2019nexus, hu2021scrooge, crankshaw2020inferline, crankshaw2017clipper} can \textit{not} minimize the serving cost due to three limitations: (1) \textit{Excessive module latency.} Existing systems dispatch requests among machines in individual request and form batched requests at the machine side, which has a relatively low batch collection rate and unnecessarily increases the module latency. (2) \textit{Low module throughput.} Existing systems only support a fixed amount of configurations for each module and have limited support for resource heterogeneity, leading to relatively low throughput under the latency budget. (3) \textit{Wasted latency budget.} Given the latency constraint for the entire application, existing systems either use quantized interval~\cite{shen2019nexus} or machine throughput~\cite{hu2021scrooge, crankshaw2020inferline} to split the latency into per-module latency budget, leading to sub-optimal result. Therefore, these limitations result in low module throughput and eventually a higher serving cost for existing systems.



In this paper, we argue that \textbf{a proper DNN inference system should dispatch requests with higher batch collection rate, support flexible module configurations and assign each module a decent latency budget}, so as to maximize module througput and thereby minimize the serving cost. Based on this, we design and implement a DNN inference system called \sysname, which provides cost-minimum DNN inference services with efficient request dispatching, module scheduling and latency splitting capability as follows.

First, \sysname designs a novel request dispatch method to minimize the latency of DNN modules by dispatching requests among machines in batched requests directly, which largely increases the batch collection rate and thereby reduces the module latency. By doing so, \sysname is able to select module configurations with larger throughput under the latency constraint, leading to a lower serving cost.


Second, \sysname leverages a novel module scheduling method to maximize module throughput. It uses a combination of batching and heterogeneous hardware to derive a multi-tuple configuration per module, which increases resource efficiency. Besides, for machines with small workload, \sysname provides dummy generator and latency reassigner to increase their cost efficiency, which further reduces the cost.


Third, \sysname proposes a novel latency splitting method to minimize the total cost for multi-DNN applications. It leverages the latency-cost efficiency to measure the amount of cost that each module configuration can save per unit of latency budget, and designs a heuristic to split the latency to minimize the total cost for the entire application. Besides, \sysname supports various algorithms to optimize the splitting results and thereby further minimize the cost.


\sysname is completed implemented as a containerized system deployed on a cluster of $16$ GPUs with a total of \texttt{23k} lines of Python code. Evaluation shows that \sysname is able to minimize the serving cost while satisfying the latency objective. Compared to \sysname, existing systems~\cite{shen2019nexus, hu2021scrooge, crankshaw2020inferline, crankshaw2017clipper} require an average of $49.3\%$ to $137.2\%$ extra serving cost. For certain workload, existing systems require up to $3.2$ times the cost of \sysname's. For each workload, we generate the optimal cost with brute force search and show that \sysname derives the optimal one for more than $91\%$ workloads with an average runtime of $5$ milliseconds. In the ablation study, we quantify the importance of \sysname's design choices.


\section{Background and Motivation}
\label{sec:background}
In this section, we start with the scheduling policies of existing serving systems, followed by three design dimensions that distinguish \sysname from prior work.



\begin{figure}[t]
    \centering
    \includegraphics[width=0.27\textwidth]{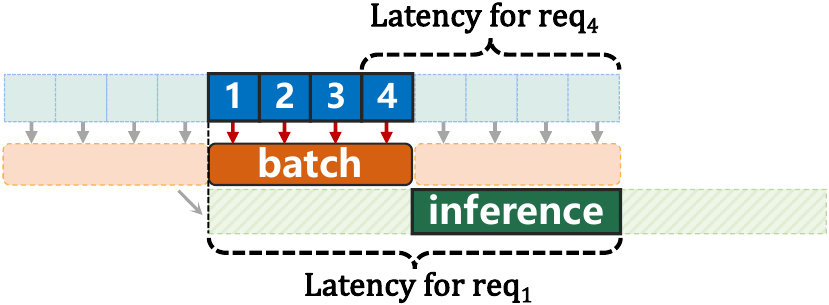}
    \caption{Requests in the same batch have varying latency.}
    \label{fig:wcl_illustration}
    \tight
\end{figure}

\begin{figure}[t]
    \centering
    \subfloat{\includegraphics[width=0.18\textwidth]{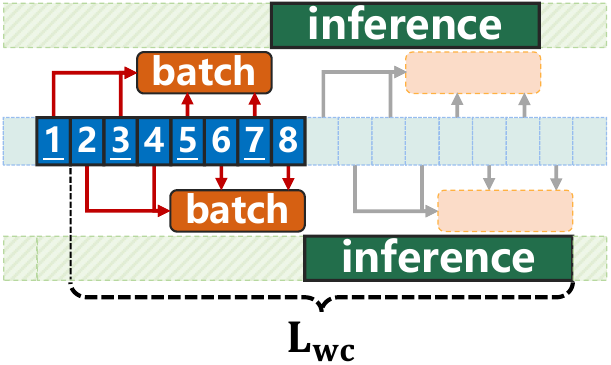}}
    \hspace{2em}
    \subfloat{\includegraphics[width=0.18\textwidth]{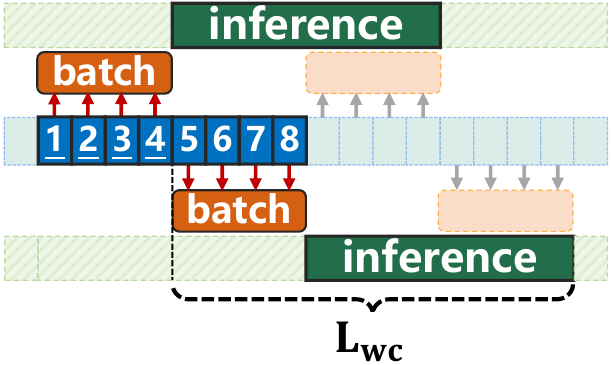}}
    \caption{(a) round-robin dispatch; (b) batch dispatch.}
    \label{fig:dispatch_problem}\tight
\end{figure}




\parab{How do existing systems decide module scheduling?} Two key factors affect module scheduling. First, batching is widely used to increase module throughput with efficient utilization of GPU's computation capability~\cite{crankshaw2017clipper, shen2019nexus, wu2023faasbatch}. Requests in the same batch experience different latency. For example, as shown in \figref{fig:wcl_illustration}, \codesub{req}{1} has a much larger latency than \codesub{req}{4}. The serving system defines the maximum latency of all requests as the worst case latency \codesub{L}{wc} and guarantees that \codesub{L}{wc} is within the target latency objective. Second, among the available computation hardware in the cloud, the most cost-efficient hardware is module dependent~\cite{crankshaw2020inferline, narayanan2020heterogeneity}. Therefore, the serving system needs to find the optimal module configuration (\eg batch size and computation hardware) to minimize the serving cost, while satisfying the latency objective.

To do so, existing systems~\cite{crankshaw2017clipper, shen2019nexus, crankshaw2020inferline, hu2021scrooge} often rely on a two-round heuristic: (1) it first estiamtes the worst case latency \codesub{L}{wc} of all candidate configurations, and (2) then greedily select the optimal one that maximize the throughput while its \codesub{L}{wc} is within the latency budget. For example, we consider a workload of module \codesub{M}{1} from \tabref{tab:batch_example} with request rate of $100$ req/sec and latency SLO of $0.4$ sec. Existing systems dispatches individual requests among machines to form batches in a round-robin fashion as shown in \figref{fig:dispatch_problem}(a), leading to a worst case latency two times of module duration (\ie $\codesub{L}{wc}=\texttt{2d}$)~\cite{shen2019nexus,hu2021scrooge}. Therefore, \codesub{L}{wc} for batch size of $2$, $4$ and $8$ will be $0.32$, $0.4$ and $0.64$ sec, and existing systems will select batch size of $4$, since it maximize module throughput while satisfying latency objective.



\begin{table}[t]
    \caption{
    Module profiles, where \texttt{\textup{b}} and \texttt{\textup{d}} are batch size and execution duration (sec), \texttt{\textup{t}=\textup{b}/\textup{d}} is module throughput (req/sec).}
    \label{tab:batch_example}
    \scriptsize
    \centering
    \begin{tabular}{ ccccccccc }
     \toprule
     \multicolumn{3}{c}{Module \codesub{M}{1}} &
     \multicolumn{3}{c}{Module \codesub{M}{2}} &
     \multicolumn{3}{c}{Module \codesub{M}{3}} \\ \cmidrule(lr){1-3}
     \cmidrule(lr){4-6}
     \cmidrule(lr){7-9}
     \texttt{b} & \texttt{d} & \texttt{t} & \texttt{b} & \texttt{d} & \texttt{t} & \texttt{b} & \texttt{d} & \texttt{t} \\
     \toprule
     2 & 0.160 & 12.5 & 2 & 0.125 & 16 & 2 & 0.100 & 20 \\
     4 & 0.200 & 20 & 4 & 0.160 & 25 & 8 & 0.250 & 32 \\
     8 & 0.320 & 25 & 8 & 0.250 & 32 & 32 & 0.800 & 40\\
     \bottomrule
    \end{tabular}
\end{table}





\parab{Request dispatching.} We argue that the request dispatching policy of existing systems cannot minimize the worst case latency \codesub{L}{wc}. As shown in \figref{fig:dispatch_problem}(a), existing systems~\cite{crankshaw2017clipper, shen2019nexus, crankshaw2020inferline, hu2021scrooge} dispatch requests one by one to each machine (\eg \codesub{req}{1,3,5,7}) and form batched request at the machine side, so each machine receives requests at a rate equivalent to \textit{its per-machine module throughput}, leading to $\codesub{L}{wc}=\texttt{2d}$.


However, if the serving system can dispatch requests among machines in batches directly, \codesub{L}{wc} can be largely reduced. As shown in \figref{fig:dispatch_problem}(b), if we dispatch a consecutive amount of requests equal to its batch size to each machine (\eg \codesub{req}{1-4}), each machine will receive requests at a rate equivalent to \textit{the total request rate}, leading to $\codesub{L}{wc}=\texttt{d}+\texttt{b/T}$, where \texttt{b}, \texttt{d} and \texttt{T} is the batch size, execution duration and total request rate.

By dispatching requests in batch, the serving system shifts the batch collection process from the machine to the frontend, which can potentially reduce the serving cost with a smaller \codesub{L}{wc}. For example, in the above example of \codesub{M}{1}, when dispatching in batches, \codesub{L}{wc} for batch size of $2,$ $4$ and $8$ will be $0.18$, $0.24$ and $0.4$ sec. Now the serving system can choose batch size $8$, which is infeasible for existing systems. Therefore, to handle $100$ req/sec, serving systems with batch-aware dispatch only require $\texttt{n}=\texttt{T}/\texttt{t}=100/25=4$ machines with batch size $8$, while existing ones with round-robin dispatch require $\texttt{n}=100/20=5$ machines with batch size $4$.

\parae{Key question \mynum{182}:} Since the incoming request rate doesn't have to be an integer multiple of the module throughput of any given configuration, each module is usually associated with multiple configurations, as we will discuss next. Therefore, for each module, \sysname needs to derive a proper request dispatch strategy across its configurations to minimize module's worst case latency.

\parab{Module scheduling.} We argue that the module scheduling policy of existing systems cannot maximize module throughput under the latency budget. As discussed, existing systems~\cite{crankshaw2017clipper, shen2019nexus, crankshaw2020inferline, hu2021scrooge} leverages a two-round greedy heuristic. 
The scheduler first chooses the optimal configuration \codesub{c}{opt}, which derives the largest module throughput \codesub{t}{opt} among all candidates while satisfying the latency objective. It allocates $\texttt{n}=\lfloor \texttt{T}/\codesub{t}{opt}\rfloor$ machines running at \codesub{c}{opt} for the majority workload $\codesub{T}{maj}=\texttt{n}\cdot \codesub{t}{opt}$. Then, the scheduler chooses another configuration \codesub{c}{res} for the residual workload $\codesub{T}{res}=\texttt{T}-\texttt{n}\cdot \codesub{t}{opt}$ and allocates additional machine accordingly. Therefore, existing systems usually generate a two-tuple $\langle \codesub{c}{opt},\codesub{c}{res}\rangle$ for each module.

\begin{table}[t]
    \setlength{\tabcolsep}{3pt}
    \scriptsize
    \centering
    \caption{Scheduling results and corresponding serving cost of four scheduling methods \textnormal{\codesub{S}{1} - \codesub{S}{4}} for \textnormal{\codesub{M}{3}}. Each scheduling result includes \textnormal{\texttt{K}} configurations, where the \textnormal{\texttt{i}th} configuration \textnormal{\codesub{T}{i} $(\codesub{n}{i}\otimes \codesub{b}{i})$} means that \textnormal{\codesub{n}{i}} machine with batch size of \textnormal{\codesub{b}{i}} will deal with request rate of \textnormal{\codesub{T}{i} req/sec}, and the cost \textnormal{$=\sum_1^\texttt{K}\codesub{n}{i}$.}}
    \label{tab:scheduling_example}
    \begin{tabular}{ccccc}
    \toprule
    Method & \codesub{S}{1} & \codesub{S}{2} & \codesub{S}{3} & \codesub{S}{4}\\
    \midrule
    Dispatch & round-robin & batch-aware & batch-aware & batch-aware \\
    \#Config & 2 & 2 & any & any \\
    Dummy & \ding{55} & \ding{55} & \ding{55} & \ding{51} \\
    \midrule
    Config & $192$ $(6.0\otimes 8)$ & $160$ $(4.0\otimes 32)$ & $160$ $(4.0\otimes 32)$ & $200$ $(5.0\otimes 32)$ \\
    & $\ \ \ \ 6$ $(0.3\otimes 2)$ & $38$ $(1.9\otimes 2)$ & $32$ $(1.0\otimes 8)$ & \\
    & & & $\ \ 6$ $(0.3\otimes 2)$ & \\
    \midrule
    Cost & 6.3 & 5.9 & 5.3 & 5.0 \\
    \bottomrule
    \end{tabular}
\end{table}


For example, assuming a workload of module \codesub{M}{3} from \tabref{tab:batch_example} with request rate of $198$ req/sec and latency SLO of $1.0$ sec, \tabref{tab:scheduling_example} shows the scheduling results and corresponding serving costs for various scheduling methods: (1) Baseline \codesub{S}{1}. Existing systems use round-robin dispatch and two-tuple configuration as \codesub{S}{1}, which requires $6$ machines at full capacity and $1$ machine at partial capacity, leading to a serving cost of $6.3$ machines. (2) \codesub{S}{1} to \codesub{S}{2}. A better solution \codesub{S}{2}, which dispatches batched requests, reduces the serving cost from $6.3$ to $5.9$ machines by reducing \codesub{L}{wc} of candidate configurations. (3) \codesub{S}{2} to \codesub{S}{3}. For certain workload, the serving cost can be further reduced with multi-tuple configurations. In the above example, \codesub{S}{3} keeps the majority workload unchanged and further splits the partial workload (\ie $38\rightarrow32+6$) to assign $2$ machines with batch size $8$ and $2$ respectively, reducing the serving cost from $5.9$ to $5.3$ machines. Existing systems can \textit{not} support multi-tuple configurations due to runtime complexity~\cite{shen2019nexus} or problem formulation~\cite{crankshaw2017clipper, hu2021scrooge, crankshaw2020inferline, romero2021infaas}. (4) \codesub{S}{3} to \codesub{S}{4}. We can surprisingly further reduce the serving cost by adding a proper amount of dummy requests. In the above example, if we add a dummy $2$ req/sec (\ie $198\rightarrow200$ req/sec), \codesub{S}{4} reduces the serving cost from $5.3$ to $5.0$ machines. Specifically, though the dummy requests will consume a certain amount of computation resources, the serving cost can be reduced by increasing the batch size of those machines assigned to the residual workload.

\parae{Key question \mynum{183}:} It is challenging to schedule modules in a complex configuration space. For example, naively adding dummy of $10$ req/sec will only add extra workload without any cost reduction. Therefore, \sysname needs to properly design module scheduling.






\parab{Latency splitting.} Recent applications are starting to use multiple DNN modules to improve performance. For such application, the serving system is responsible to split the end-to-end latency into per-module latency budget. We argue that existing systems cannot utilize the latency budget efficiently.

Deriving the optimal latency splitting strategy is known to be NP-Complete~\cite{shen2019nexus, hu2021scrooge, crankshaw2020inferline}. Therefore, existing systems mainly leverage two types of heuristics. The first one~\cite{shen2019nexus} divides the end-to-end latency into discretized intervals to reduce the search space. However, its optimality is decided by the interval length and its complexity is exponential to the number of modules and the reciprocal of interval length. The second one~\cite{hu2021scrooge, crankshaw2020inferline} allocates latency budget across modules greedily according to module throughput. However, it prefers modules with larger throughput without consideration for its efficiency on latency budget, leading to sub-optimal solutions (\secref{sec:eval}).





\parae{Key question \mynum{184}:} Utilizing the end-to-end latency budget is critical to minimize the total cost for multi-DNN application. Therefore, \sysname needs to design a proper latency splitting strategy to derive per-module latency budget.


\section{Design}
\label{sec:design}
In this section, we start with an overview of \sysname, followed by a detailed description of its components.
\vspace{-3pt}




\subsection{Overview}
\label{sec:system_overview}
\parab{Terminology.} Similar to existing works~\cite{shen2019nexus, hu2021scrooge}, we define all requests for the same DNN-based application as a session. Each session has a unique \textit{session id}, associated with (1) an application directed acyclic graph (DAG), where a node in the DAG represents a DNN module or processing module for GPU/CPU computation, and an edge represents computation dependency across nodes; (2) the request rate for each node in the DAG; and (3) a latency objective, which is the expected end-to-end latency for the entire application.

\begin{figure}[t]
    \centering
    \includegraphics[width=0.46\textwidth]{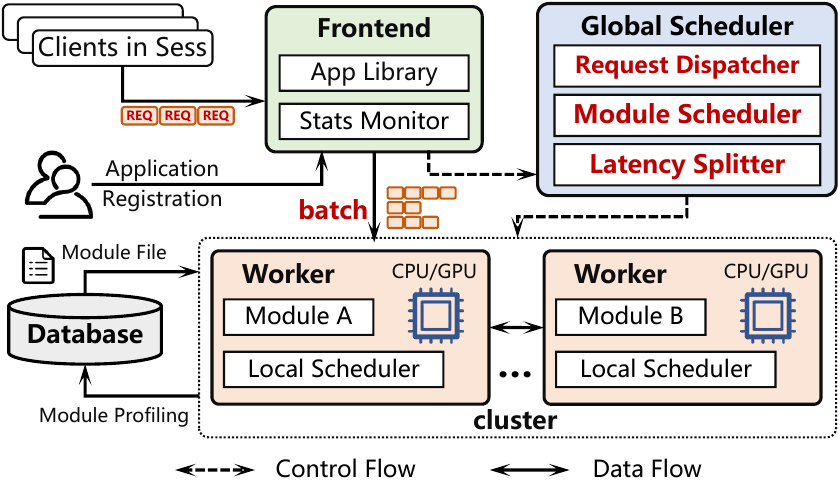}
    \caption{\sysname overview.}
    \label{fig:system_overview}
    \tight
\end{figure}


To capture the resource usage of modules, \sysname provides a profiling library in the shared database for each module, which includes the execution duration of the given module under various configurations (\eg batch size and computation hardware). Similar to existing works~\cite{crankshaw2017clipper, shen2019nexus, crankshaw2020inferline, hu2021scrooge, romero2021infaas, hu2021rim}, the profiling is collected offline once when the application is registered, and will \textit{not} affect the runtime latency of requests.

The cost is defined as the actual amount of computation resources used by the workload, according to frame-rate proportionality~\cite{hu2021scrooge}. Specifically, for a session with \texttt{m} modules, if the scheduler allocates \codesub{N}{i} machines to handle module \codesub{M}{i}, the session will have a cost of $\sum_{\texttt{i}=\texttt{1}}^\texttt{m}\sum_{\texttt{j}=\texttt{1}}^{\codesub{N}{i}}\codesub{p}{ij}\cdot\codesub{f}{ij}/\codesub{t}{ij}$, where \codesub{p}{ij}, \codesub{f}{ij} and \codesub{t}{ij} is the unit price, the assigned request rate and the module throughput for the \texttt{j}th machine of \codesub{M}{i}.






\parab{How \sysname minimizes the serving cost.} As shown in \figref{fig:system_overview}, for each session, the global scheduler decides the configuration of each module given its request rate and latency objective. Request dispatcher first derives the proper batch-aware dispatching strategy to minimize \codesub{L}{wc}. Module scheduler and latency splitter then iteratively update candidate module configuration and per-module latency budget to minimize the total serving cost for the entire multi-DNN application.

Next, we first describe \sysname's request dispatching method (\secref{sec:dispatching}), followed by its module scheduling method (\secref{sec:scheduling}), and finally the latency splitting method (\secref{sec:splitting}).

\subsection{Request Dispatching}
\label{sec:dispatching}
To serve DNN-based applications, the first step is to estimate the worst case latency \codesub{L}{wc} of all candidate configurations. As discussed, \codesub{L}{wc} is closely related to request dispatch policy, and existing ones cannot minimize \codesub{L}{wc}. \sysname addresses the issue with its novel request dispatcher.

\parab{Throughput-Cost Request Dispatcher.} Each candidate configuration is actually a set of configurations to handle a given module's majority and residual workload. For each set, \sysname will rank all machines according to their configuration's throughput-cost ratio \texttt{r}, which is the module throughput divided by the hardware unit price. For example, for an incoming workload of $8$ req/sec for module \codesub{M}{4}, one of the candidate configuration set includes three machines \texttt{A}, \texttt{B} and \texttt{C}, where both \texttt{A} and \texttt{B} use configuration of batch size $6$ with execution duration of $2.0$ sec, and \texttt{C} uses configuration of batch size $2$ with execution duration of $1.0$ sec. All three machines have the same hardware unit price of $1.0$. \sysname calculates their throughput-cost ratio as $\codesub{r}{A}=\codesub{r}{B}=\frac{\texttt{b}}{\texttt{d}}/\texttt{p}=\frac{6}{2.0}/1.0=3.0$ and $\codesub{r}{C}=\frac{2}{1.0}/1.0=2.0$, where \texttt{b}, \texttt{d} and \texttt{p} is the corresponding batch size, execution duration and hardware unit price. Therefore, \sysname will rank these three machines as $\codesub{r}{A}=\codesub{r}{B}>\codesub{r}{C}$.


\begin{figure}[t]
    \centering
    \includegraphics[width=0.48\textwidth]{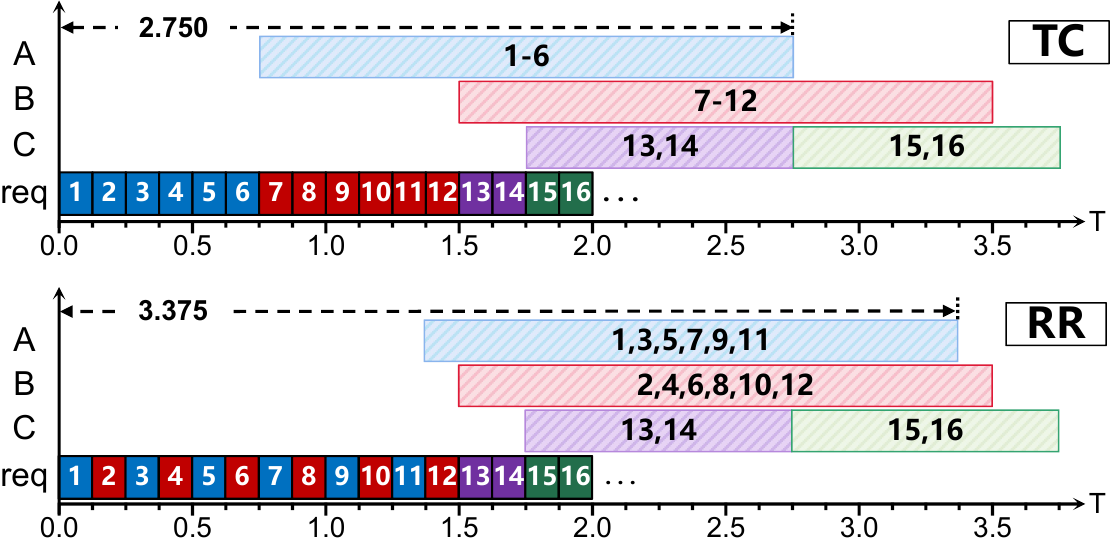}
    \caption{Request dispatching results for throughput-cost dispatch policy (top) and round-robin dispatch policy (bottom).}
    \label{fig:tc_dispatch}
\end{figure}


\sysname dispatches batched requests among machines in the order of throughput-cost ratio, where each machine will receive a successive amount of requests equal to its batch size in a row. For machines with the same throughput-cost ratio (\eg machines that use the same configuration such as \texttt{A} and \texttt{B} in the above example), \sysname will dispatch requests among them in batched requests in turn. On the contrary, existing systems~\cite{crankshaw2017clipper, shen2019nexus, hu2021scrooge} dispatch requests among machines as individual request in round robin and let each machine collect the batch locally. We call \sysname's dispatching strategy throughput-cost (TC) dispatch policy, as opposed to existing systems' round-robin (RR) dispatch policy. As discussed in \secref{sec:background}, TC dispatch policy ensures that each machine collects requests to form the batch at the rate of the whole workload, while RR dispatch policy can only collect requests at the rate of machine's own module throughput.

By accelerating the batch collection rate, the TC dispatch policy can achieve a lower worst case latency than the RR dispatch policy. For example, given the above example of $8$ req/sec for module \codesub{M}{4}, \figref{fig:tc_dispatch} shows the request dispatching results using TC and RR dispatch policy, respectively. For the first $16$ requests, the TC dispatch policy will dispatch \codesub{req}{1-6} to \texttt{A}, \codesub{req}{7-12} to \texttt{B} and \codesub{req}{13-16} to \texttt{C}, leading to a worst case latency of $2.75$ sec with $0.75$ sec for batch collection. Meanwhile, the RR dispatch policy will dispatch requests among \texttt{A} and \texttt{B} back and forth for \codesub{req}{1-12} without considering batch information, and then send \codesub{req}{13-16} to \texttt{C}, leading to a worst case latency of $3.375$ sec with $1.375$ sec for batch collection. 



For a module, given a configuration set of \texttt{n} machines ordered by throughput-cost ratio: $\codesub{r}{1}\ge\codesub{r}{2}\ge\cdots\ge\codesub{r}{n}$, we denote the configuration of machine \texttt{i} as \codesub{c}{i} and the assigned request rate for machine \texttt{i} as \codesub{f}{i}. The remaining workload for machine \texttt{i} is defined as $\codesub{w}{i}=\sum_{\texttt{j}=\texttt{1}}^\texttt{n}\codesub{s}{ij}\cdot\codesub{f}{j}$, where \codesub{s}{ij} is $1$ if $\codesub{r}{j}\le\codesub{r}{i}$ and is $0$ if $\codesub{r}{j}>\codesub{r}{i}$. Namely, \codesub{w}{i} represents the amount of request rate assigned to machines with throughput-cost ratio no larger than machine \texttt{i}'s. For example, in the above example of \codesub{M}{4}, the remaining workload for both \texttt{A} and \texttt{B} is $3+3+2=8$ req/sec, while the one for \texttt{C} is $2$ req/sec.


\begin{theorem}
Using TC dispatch, the worst case latency for machine \texttt{i} is $\codesub{L}{wc}(\texttt{i})=\codesub{d}{i}+\codesub{b}{i}/\codesub{w}{i}$, where \codesub{d}{i} and \codesub{b}{i} is the execution duration and batch size under configuration \codesub{c}{i}. The worst case latency for the module is $\max_{\texttt{i}=\texttt{1}}^\texttt{n}\codesub{L}{wc}(\texttt{i})$.
\label{theorem:wcl_bound}
\end{theorem}

\begin{proof}
TC dispatch policy sends batched requests among machines in the order of throughput-cost ratio, so machine \texttt{i}'s remaining workload \codesub{w}{i} is its equivalent batch collection rate, leading to a batch collection time $\codesub{L}{col}=\codesub{b}{i}/\codesub{w}{i}$. Therefore, $\codesub{L}{wc}(\texttt{i})$ will be $\codesub{d}{i}+\codesub{b}{i}/\codesub{w}{i}$, and the module's worst case latency will be the maximum of each machine's $\codesub{L}{wc}(\texttt{i})$.
\end{proof}\tight



As shown in \theoremref{theorem:wcl_bound}, \sysname's TC dispatch can achieve a worst case latency of $\codesub{d}{i} + \codesub{b}{i} / \codesub{w}{i}$ for each machine \texttt{i}, which is smaller than existing systems' $\texttt{2}\codesub{d}{i}$ worst case latency using RR dispatch. By minimizing \codesub{L}{wc} for all candidate configurations, \sysname is able to choose configurations with larger throughput, which is infeasible for existing systems due to latency constraints, leading to a smaller serving cost (\secref{sec:eval}).

\subsection{Module Scheduling}
\label{sec:scheduling}
Next, the DNN serving system needs to schedule modules given its latency budget and request rate. As discussed, existing systems use two-tuple configuration and cannot achieve high module throughput. \sysname addresses the issue with its novel module scheduler and residual optimizer.






\begin{algorithm}[t]
\scriptsize
\caption{Generating the configuration set for module $M$}
\label{algo:bpc_generator}
\begin{algorithmic}[1]
\Function{GenerateConfig}{$T_M$, $L_M$, $P_M$}\label{ln:gc_define}
\State $\text{rw}\leftarrow T_M$, $\text{configs}\leftarrow []$\label{ln:rw_initial}
\State $k\leftarrow0$, $c\leftarrow P_M[k]=\langle b_k,s_k,d_k,t_k,hw_k,p_k\rangle$\label{ln:k_initial}
\While{$\text{rw}\ne0$}
    \If{$\text{GetWCL}(c)\le L_M$}\label{ln:check_slo}
        \State $n\leftarrow\text{rw}/t$
        \If{$n\ge1$}
            \State $\text{configs}\leftarrow\text{configs}\oplus\langle c,\lfloor n\rfloor\rangle$\label{ln:allocate_full}
            \State $\text{rw}\leftarrow\text{rw}$ - $\lfloor n\rfloor\times t$\label{ln:update_rw}
        \Else
            \State $\text{configs}\leftarrow\text{configs}\oplus\langle c,n\rangle$\label{ln:allocate_partial}
            \State $\text{rw}\leftarrow0$\label{ln:zero_rw}
        \EndIf
    \Else
        \While{$\text{GetWCL}(c)>L_M$}\label{ln:find_next_1}
            \State $c\leftarrow P_M[\text{++}k]$\label{ln:find_next_2}
            \If{$k\ge\text{len}(P_M)$}\label{ln:cant_find_1}
                \State \Return False, []\label{ln:cant_find_2}
            \EndIf
        \EndWhile
    \EndIf
\EndWhile
\State \Return True, configs
\EndFunction
\end{algorithmic}
\end{algorithm}



\parab{Module Scheduler.} Given the worst case latency of all candidate configurations, \sysname derives the multi-tuple configuration set for module \texttt{M} using \algoref{algo:bpc_generator}, where the input is the request rate \codesub{T}{M}, the latency budget \codesub{L}{M} and the profiling \codesub{P}{M} for all candidate configurations ordered by throughput-cost ratio (\lineref{ln:gc_define}). \sysname maintains the current unallocated workload \texttt{rw}, starting at \codesub{T}{M} (\lineref{ln:rw_initial}). The configuration index \texttt{k}, starting at $0$, records the current configuration \texttt{c} (\lineref{ln:k_initial}). Function \texttt{GetWCL()} estimates \codesub{L}{wc} as described in \theoremref{theorem:wcl_bound}. \algoref{algo:bpc_generator} generates the configuration set greedily. If the current configuration \texttt{c} can satisfy the latency budget (\lineref{ln:check_slo}), \sysname will use \texttt{c} to update \texttt{configs}. When the number of machines \texttt{n} is larger than $1$, $\lfloor\texttt{n}\rfloor$ machines will run at \texttt{c} at full capacity (\lineref{ln:allocate_full}) and corresponding amount of workload will be deducted from \texttt{rw} (\lineref{ln:update_rw}). Otherwise, one machine will run at \texttt{c} at partial capacity of \texttt{n} (\lineref{ln:allocate_partial}) and set \texttt{rw} to $0$ (\lineref{ln:zero_rw}). If \texttt{c} can \textit{not} satisfy the latency objective, \sysname will find the next feasible configuration (\multilineref{ln:find_next_1}{ln:find_next_2}).






\parab{Optimizing Residual Workload.} \sysname leverages its dummy generator and latency reassigner to increase module throughput for residual workload as follows.

\parae{Dummy generator} adds a proper amount of dummy requests to the given module. Given a module with \texttt{K} distinct module configurations ordered by throughput-cost ratio: $\frac{\codesub{t}{1}}{\codesub{p}{1}}>\frac{\codesub{t}{2}}{\codesub{p}{2}}>\cdots>\frac{\codesub{t}{K}}{\codesub{p}{K}}$, where the \texttt{i}th module configuration \codesub{c}{i} with module throughput \codesub{t}{i} and unit price \codesub{p}{i} is assigned \codesub{n}{i} machines, \sysname defines the \textit{leftover workload} for \codesub{c}{i} as $\codesub{u}{i}=\sum_{\texttt{j=i+1}}^\texttt{K}\codesub{n}{j}\codesub{t}{j}$. Different from the remaining workload defined in \secref{sec:dispatching}, the leftover workload represents the total amount of requests assigned to machines with throughput-cost ratio less than the given module configuration's.

\begin{theorem}
In the cost-minimum configuration, the leftover workload \codesub{u}{i} for the \texttt{i}th configuration \codesub{c}{i} ordered by throughput-cost ratio should be less than its throughput \codesub{t}{i}. Namely, $\forall \texttt{i} \in \mathbb{K}, \codesub{u}{i} < \codesub{t}{i}$.
\label{theorem:leftover_workload}
\end{theorem}



\begin{proof}
We use proof of contradiction, where $\exists k \in \mathbb{K}, u_k=\sum_{j=k+1}^K n_jt_j \ge t_k$. For $u_k$, we define the total number of machine as $n_{\alpha}=\sum_{j=k+1}^K n_j$, the average unit price as $p_{\alpha}=\sum_{j=k+1}^K n_jp_j/n_{\alpha}$, and the average throughput as $t_{\alpha}=u_k/n_{\alpha}$. The total cost for $u_k$ is $C_0=\sum_{j=k+1}^K n_jp_j$. Since $u_k\ge t_k$, we allocate a new machine at $c_k$ to handle $t_k$ workload for higher throughput-cost. Depending on the worst case latency for the remaining $u_k-t_k$ workload, there will be two cases.

In the first case, the original configuration for $u_k-t_k$ can still satisfy the latency constraint, then \sysname will switch $t_k$ amount of workload from its original configuration to $c_k$ for higher throughput-cost efficiency and keep $u_k-t_k$ workload unchanged. We define $T_1=\lfloor u_k/t_k\rfloor\cdot t_k$, then the new total cost for $u_k$ is $C_1=p_k\cdot\lfloor u_k/t_k\rfloor+p_{\alpha}\cdot(u_k-T_1)/t_{\alpha}$. We have $C_0-C_1=T_1(p_{\alpha}/t_{\alpha}-p_k/t_k)$. According to definition, $t_k/p_k>t_{\alpha}/p_{\alpha}$, so $C_0-C_1>0$. Namely, the cost-minimum assumption is violated.

In the second case, the original configuration for $u_k-t_k$ can not satisfy the latency constraint, then \sysname will choose a new configuration $c_{\beta}$ to deal with the remaining $u_k-t_k$ workload to guarantee the latency SLO. Similarly, the new total cost for $u_k$ is $C_2=p_k\cdot \lfloor u_k/t_k\rfloor+p_{\beta}\cdot(u_k-T_1)/t_{\beta}$. We have $C_0-C_2=T_1(\frac{p_{\alpha}}{t_{\alpha}}-\frac{p_k}{t_k})-(u_k-T_1)(\frac{p_{\beta}}{t_{\beta}}-\frac{p_{\alpha}}{t_{\alpha}})$. According to definition, $T_1>u_k-T_1$, and for most modules, $\frac{p_{\alpha}}{t_{\alpha}}-\frac{p_k}{t_k} \approx \frac{p_{\beta}}{t_{\beta}}-\frac{p_{\alpha}}{t_{\alpha}}$, so $C_0-C2>0$. Namely, the cost-minimum assumption is violated.
\end{proof}\tight


Given \theoremref{theorem:leftover_workload}, configuration \codesub{c}{i}'s leftover workload requires up to one extra machine at \codesub{c}{i}. Therefore, for a given module \texttt{M} with latency budget \codesub{L}{M}, \sysname determines whether adding dummy request of $\codesub{dum}{i}=\codesub{t}{i}-\codesub{u}{i}$ can reduce the serving cost. For example, in the example of \codesub{M}{3} with request rate of $198$ req/sec in \secref{sec:background}, according to the definition, \codesub{u}{i} for module configuration with batch size $32$ should be $\codesub{u}{i}=32+6=38$ req/sec. A dummy request of $\codesub{dum}{i}=40-38=2$ req/sec will reduce the serving cost to $5.0$ machines, corresponding to scheduling result \codesub{S}{4} in \tabref{tab:scheduling_example}.

\parae{Latency reassigner} reassigns the remaining latency budget to module's residual workload. \sysname derives the configurations using \algoref{algo:bpc_generator}. However, there can still be a gap between the worst case latency and the latency budget, which can be used to further reduce the serving cost when reassigned properly. Intuitively, this latency gap will \textit{not} benefit the majority configuration for either module, otherwise \algoref{algo:bpc_generator} should output another configuration result in the first place. Instead, reassigning the latency gap to module's residual workload can potentially increase its throughput. \sysname re-runs \algoref{algo:bpc_generator} with the updated latency budget for residual workload, and updates the configuration set if doing so can reduce the serving cost.




The intuition behind these two methods to increase the throughput for residual workload is associated with the latency constraint: $\texttt{d}+\texttt{b}/\texttt{w}\le \codesub{L}{M}$. Adding dummy requests is equivalent to increasing the remaining workload \texttt{w}, while reassigning the latency budget is equivalent to increasing the latency budget \codesub{L}{M} for the residual workload, both of which enable \sysname to choose configurations with larger throughput for the residual workload, leading to a lower serving cost (\secref{sec:eval}).

\subsection{Latency Splitting}
\label{sec:splitting}
Latency splitting derives the per-module latency budget. As discussed, existing ones often lead to sub-optimal results. \sysname provides heuristics based on latency-cost efficiency to gradually allocate latency budget across modules and supports splitting optimizer to further reduce the total cost.

\parab{Latency Splitter.} \sysname defines the latency-cost efficiency \texttt{LC} as the amount of cost that a given module configuration can reduce per unit of latency budget when switching from the previous module configuration \codesub{c}{prev} to the new one \codesub{c}{new}. Specifically, for a module with request rate of \texttt{T}, $\texttt{LC}=\frac{\codesub{C}{M}(\texttt{prev})-\codesub{C}{M}(\texttt{new})}{\codesub{L}{wc}(\texttt{new})-\codesub{L}{wc}(\texttt{prev})}$, where $\codesub{C}{M}(*)=\texttt{p}_*\times \texttt{T}/\texttt{t}_*$ and $\codesub{L}{wc}(*)$ is the serving cost and worst case latency for module \texttt{M} under the previous/new configuration. For example, for module \codesub{M}{1} from \tabref{tab:batch_example}, we assume that the previous configuration is batch size of $2$ to deal with a workload of $\texttt{T}=100$ req/sec and that $\codesub{p}{*}$ is $1.0$. The other two configurations with batch size of $4$ and $8$ are the candidates to switch to. According to the definition, the latency-cost efficiency for configuration with batch size of $4$ is $\frac{1.0\times 100/12.5-1.0\times 100/20}{(0.2+4/100)-(0.16+2/100)}=50.0$, while the one for configuration with batch size of $8$ is $18.2$.



\begin{algorithm}[t]
\scriptsize
\caption{Deriving per-module latency budget}
\label{algo:latency_splitter}
\begin{algorithmic}[1]
\Function{SplitLatency}{$T$, $L$, $P$}\label{ln:sl_define}
\State $\text{DAG}\leftarrow\text{GetDefaultDAG()}$\label{ln:default_dag}
\State $\text{flag}\leftarrow\text{True}$
\While{flag}\label{ln:iter_begin}
    \State $\text{flag}\leftarrow\text{False}$, $LC_{max}\leftarrow-\infty$
    \State $M_{max}\leftarrow\text{NULL}$, $c_{max}\leftarrow\text{NULL}$
    \For{$M$ in DAG}
        \State $T_M\leftarrow T[M]$, $P_M\leftarrow P[M]$
        \State $c_{prev}\leftarrow\text{DAG}[M]$
        \For{$c_{new}$ in $P_M$}
            \State $LC\leftarrow\frac{p_{prev}\times T_M/t_{prev} - p_{new}\times T_M/t_{new}}{\text{GetWCL}(c_{new}) - \text{GetWCL}(c_{prev})}$\label{ln:lc_def}
            \If{$LC>LC_{max}$}
                \If{GetLat(DAG, $M$, $c_{new}$) $\le L$}\label{ln:new_wcl_constraint}
                    \State $flag\leftarrow\text{True}$, $LC_{max}\leftarrow LC$
                    \State $M_{max}\leftarrow M$, $c_{max}\leftarrow c_{new}$\label{ln:iter_end}
                \EndIf
            \EndIf
        \EndFor
    \EndFor
    \If{$flag$}
        \State updateDAG($DAG$, $M_{max}$, $c_{max}$)\label{ln:get_lc_max}
    \EndIf
\EndWhile
\State \Return $[L_M$ for $M$ in DAG$]$
\EndFunction
\end{algorithmic}
\end{algorithm}



Based on \texttt{LC}, \sysname uses \algoref{algo:latency_splitter} to partition the latency into per-module latency budget. It uses \texttt{DAG} to store configurations for each module, starting at a default \texttt{DAG} (\lineref{ln:default_dag}), where each module has batch size of $1$ on hardware with the highest unit price, corresponding to the least cost-efficient configuration. Function \texttt{GetLat(DAG, M, c)} calculates the end-to-end latency when replacing module \texttt{M} in \texttt{DAG} with configuration \texttt{c} and keeping the rest of modules unchanged. Among all configurations of all modules, \algoref{algo:latency_splitter} chooses the one whose \texttt{LC} (\lineref{ln:lc_def}) is the maximum, while the latency constraint can be satisfied (\lineref{ln:new_wcl_constraint}), and then updates the configuration to the new one (\lineref{ln:get_lc_max}). \algoref{algo:latency_splitter} repeats the above process until no configuration can be updated while preserving the latency constraint.


In each iteration (\multilineref{ln:iter_begin}{ln:get_lc_max}), \sysname prefers configurations with the largest latency-cost efficiency over the one with the largest throughput, which does \textit{not} mean that \sysname sacrifices throughput. Instead, by preferring latency-cost efficiency over throughput, \sysname is able to gradually allocate the latency among modules in multiple iterations, which is more likely to achieve global optimal. On the contrary, throughput-based methods recklessly allocate the latency, which can easily fall into local optimal (\secref{sec:eval}).

\begin{table*}[t]
  \caption{Comparison of Serving Systems}
  \label{tab:eval_comparison}
  \scriptsize
  \centering
  \begin{tabular}{|c|c|c|c|c|c|c|c|}
    \hline
    System & Worst Case Latency & Num of Configurations & Batch & Hetero & Residual Optimization & Latency Split & Split Optimization\\
    \hline
    \sysname & $\texttt{d}+\texttt{b}/\texttt{w}$ & any & \checkmark & \checkmark & dummy + reassign & latency-cost efficiency & merge + direct\\
    \hline
    Nexus~\cite{shen2019nexus} & \texttt{2d} & $2$ & \checkmark & & & quantized interval & \\
    \hline
    Scrooge~\cite{hu2021scrooge} & $\texttt{d}+\texttt{b}/\texttt{t}$ & $2$ & \checkmark & \checkmark & & throughput-based & \\
    \hline
    InferLine~\cite{crankshaw2020inferline} & \texttt{2d} & $1$ & \checkmark & \checkmark & & throughput-based & \\
    \hline
    Clipper~\cite{crankshaw2017clipper} & \texttt{2d} & $1$ & \checkmark & & & evenly splitting & \\
    \hline
  \end{tabular}\tight
\end{table*}

\parab{Optimizing the Splitting Result.} \sysname supports two methods to further optimize the splitting process.


\parae{Node merger} merges modules sharing the same parent and children modules into a super-module and calculates their \texttt{LC} as a whole. The intuition is that modules sharing the same parent and children should have the same latency budget, so the latency splitting process should consider them as a whole. For example, we consider an application with three modules \codesub{M}{x}, \codesub{M}{y} and \codesub{M}{z}, where the output of \codesub{M}{x} is the input to \codesub{M}{y} and \codesub{M}{z}. We assume the latency budget is only sufficient to update one module and \texttt{LC} to update these three modules is $20$, $10$ and $15$. Using \algoref{algo:latency_splitter}, \sysname will update \codesub{M}{x} and keep \codesub{M}{y} and \codesub{M}{z} unchanged. However, if we merge \codesub{M}{y} and \codesub{M}{z} together, their total \texttt{LC} exceeds \codesub{M}{x}'s (\eg $10+15>20$), leading to a lower total cost.




\parae{Cost-direct} reverses \algoref{algo:latency_splitter}'s final \texttt{R} iterations and greedily selects the candidate operation that reduces the most cost. The intuition is that during the last iterations, the remaining latency budget is small, so using cost to select operation directly can achieve a lower cost. For example, we consider an application with two modules \codesub{M}{x} and \codesub{M}{y} in sequence, where the remaining latency budget is only sufficient to update one module. We assume that the candidate operation for \codesub{M}{x} will consume $0.01$ sec of latency budget with \texttt{LC} of 20, while the one for \codesub{M}{y} will consume $0.1$ sec of latency budget with \texttt{LC} of $10$. Updating \codesub{M}{y}, though having a smaller \texttt{LC} (\ie $10<20$), can further reduce the total cost (\ie $0.1\times 10>0.01\times 20$).







\section{Evaluation}
\label{sec:eval}

We compare \sysname against existing systems in terms of its cost minimization capability under latency constraints, and then perform an ablation study to quantify the importance of \sysname's design decisions on cost reduction.

\subsection{Methodology}
\parab{Implementation.} We have implemented all the features of \sysname described in \secref{sec:design} with roughly \texttt{8k} lines of Python code for its core functionality, and an additional \texttt{15k} lines for the app library. Each component of \sysname runs in container for resource isolation and elastic scalability. \sysname supports DNN modules trained by various frameworks including Tensorflow~\cite{abadi2016tensorflow} and PyTorch~\cite{paszke2019pytorch}. \sysname is deployed on a cluster with $8$ P100 GPUs and $8$ V100 GPUs to demonstrate its capability to minimize the cost for heterogeneous hardware. We provide extensible APIs to register new applications with less than 20 lines of code without modifying the code.


\parab{Workloads.} For fair comparison, we choose five multi-DNN applications used in existing systems~\cite{shen2019nexus, hu2021scrooge, crankshaw2017clipper, crankshaw2020inferline, romero2021infaas}: \texttt{traffic}~\cite{zhang2017live} uses variants of SSD~\cite{liu2016ssd} to detect vehicle and pedestrian in traffic video, \texttt{face} uses PRNet~\cite{feng2018joint} to detect facial keypoints, \texttt{pose} uses OpenPose~\cite{cao2017realtime} to recognize human poses, \texttt{caption} uses S2VT~\cite{venugopalan2015sequence} to generate text description of video streams and \texttt{actdet}~\cite{liu2019caesar} detects human activities across camera footages. To demonstrate \sysname's capability to minimize serving cost, we synthesize a total of \evalnum workloads of the above five multi-DNN application using public video streams~\cite{shen2019nexus, venugopalan2015sequence, liu2019caesar, earthcam}. 




\begin{figure}[t]
    \centering
    \subfloat{\includegraphics[width=0.23\textwidth]{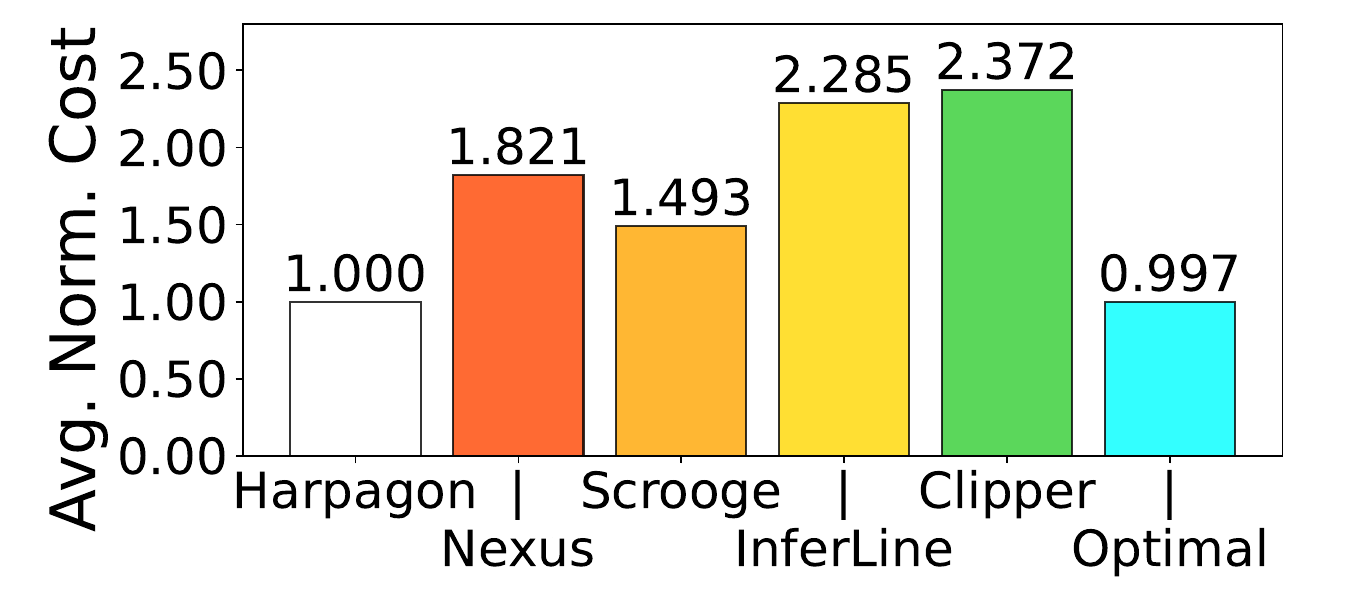}}
    \subfloat{\includegraphics[width=0.23\textwidth]{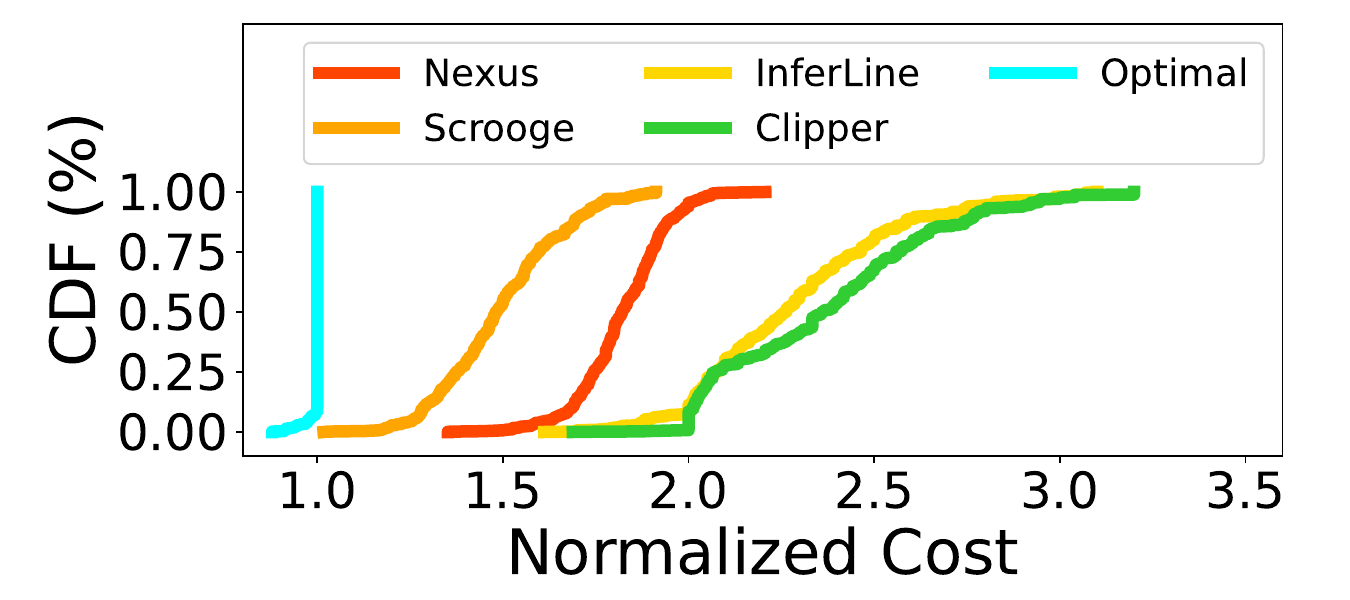}}
    \caption{(a) The average normalized cost for \sysname, four baseline systems and the optimal solution under \evalnum workloads; (b) CDF of normalized cost.}
    \label{fig:macro_result_combined}\tight
\end{figure}

\parab{Comparison Baselines.} We compare \sysname against four open-sourced baseline systems: Nexus~\cite{shen2019nexus}, Scrooge~\cite{hu2021scrooge}, InferLine~\cite{crankshaw2020inferline} and Clipper~\cite{crankshaw2017clipper}, as well as the optimal solution derived from brute force search. Similar to \sysname, all four baseline systems orchestrate DNN-based applications to minimize the cost while satisfying latency objectives. To serve multi-DNN applications, Nexus uses quantized interval to split latency across modules, Scrooge and InferLine leverage module throughput to select configuration for each module. Clipper has limited support for multi-DNN application, so we split latency across modules equally as in \cite{shen2019nexus, hu2021scrooge}.


\parab{Metrics.} We define \textit{the normalized cost} of a system as the ratio of the system's serving cost over \sysname's. The ideal normalized cost should be as small as possible.


\begin{figure*}
    \centering
    \includegraphics[width=0.98\linewidth]{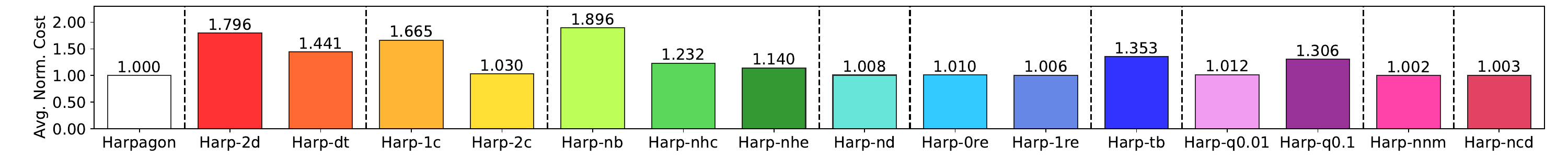}
    \caption{The average normalized cost for \sysname and baseline \sysname without corresponding functionalities.}
    \label{fig:ablation_all_a}\tight
\end{figure*}

\subsection{Comparison Results}
\label{sec:comparison_results}

\figref{fig:macro_result_combined}(a) shows the average normalized cost for \sysname, four baseline systems and the optimal solution under \evalnum workloads, while \figref{fig:macro_result_combined}(b) shows the cumulative distribution function (CDF) of normalized cost. Among five inference systems, \sysname achieves the minimum serving cost for \textit{all} workloads. Compared to \sysname, the four baseline systems require an average extra cost of $49.3\%-137.2\%$ and a maximum extra cost of $91.3\%-220.0\%$. \tabref{tab:eval_comparison} shows the differences in design choices between \sysname and the four baselines. Since the performance differences might come from a combination of multiple factors, we conjecture that the advantage of \sysname comes from the following ones.



First, the request dispatching policy of the baseline systems can \textit{not} minimize the worst case latency of candidate configurations. \sysname's TC dispatch policy achieves a lower worst case latency than baseline systems', so \sysname can leverage the saved latency to choose configurations with larger throughput or allocate more latency budget to other modules, both of which can reduce the serving cost.


Second, the module scheduling policy of the baseline systems can \textit{not} maximize the module throughput under the given latency budget. \sysname supports multiple configurations per module for higher resource efficiency, as opposed to baseline systems' fixed settings. \sysname supports hardware heterogeneity to find proper computation hardware for each module. Besides, \sysname leverages dummy generator and latency reassigner to increase module throughput of residual workload, which further reduces the serving cost.

Third, the latency splitting policy of the baseline systems can \textit{not} utilize the latency efficiently. Early serving systems~\cite{crankshaw2017clipper} have limited support for multi-DNN applications. Recent serving systems~\cite{shen2019nexus, hu2021scrooge, crankshaw2020inferline} use quantized interval or throughput-based splitting strategies, which have high runtime complexity or sub-optimal results. \sysname leverages latency-cost efficiency to gradually allocate latency budget, and supports node merger and cost-direct to further minimize the total cost.

Besides, as shown in \figref{fig:macro_result_combined}(b), compared to the optimal solution derived from brute force search, \sysname's cost is larger than the optimal for \textit{only} $8.5\%$ workloads with up to $12.1\%$ extra cost. However, the average runtime of brute force is $35.9$ sec per workload, while \sysname's runtime is \textit{only} $0.005$ sec. Namely, \sysname derives the optimal for $91.5\%$ workloads, while being more than $7000$ times faster.


\subsection{Ablation Study}
\label{sec:ablation}
In this section, we disable each of \sysname's novel features to quantify its capability on cost minimization.


\parab{The importance of TC dispatch.} \sysname designs TC dispatch to minimize the worst case latency \codesub{L}{wc} for all candidate configurations. To verify its capability to minimize the cost, we compare \sysname against: (1) \sysabbr-$2$d, which dispatches requests as individual ones~\cite{shen2019nexus, crankshaw2017clipper, crankshaw2020inferline}, and (2) \sysabbr-dt, which dispatches requests at rate of machine throughput~\cite{hu2021scrooge}.

\begin{figure}[t]
    \centering
    \subfloat{\includegraphics[width=0.235\textwidth]{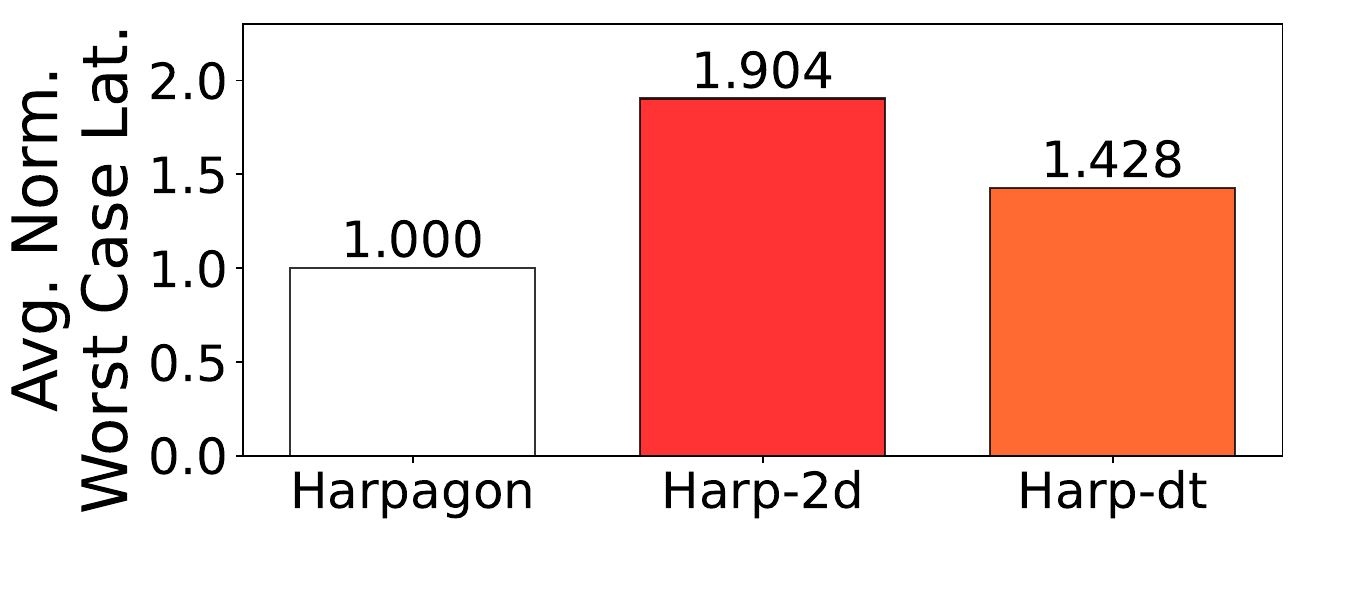}}
    \subfloat{\includegraphics[width=0.235\textwidth]{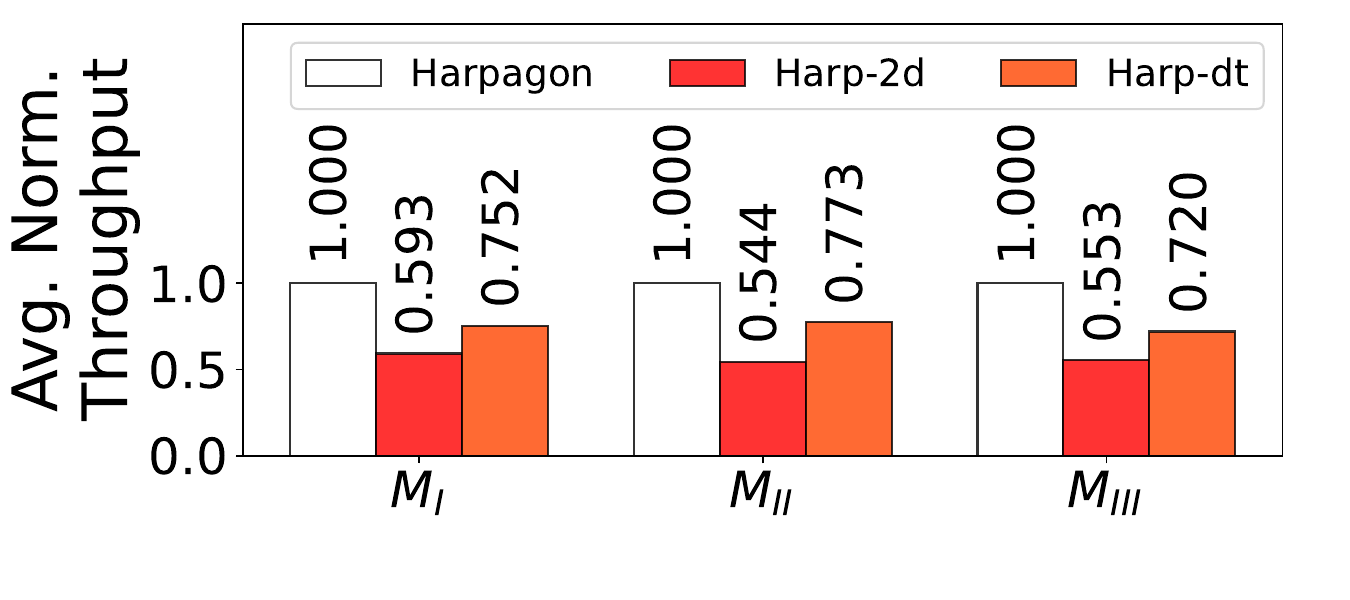}}
    \caption{(a) The average norm. worst case latency for \sysabbr-$2$d and \sysabbr-dt; (b) The average norm. throughput of three given modules.}
    \label{fig:ablation_1_combined}\tight
\end{figure}


\figref{fig:ablation_all_a} includes the normalized cost for \sysabbr-$2$d and \sysabbr-dt, which require $79.6\%$ and $44.1\%$ extra cost on average. As discussed in \secref{sec:dispatching}, TC dispatch minimizes \codesub{L}{wc} with higher batch collection rate. To verify that, we choose configurations derived from \sysabbr-$2$d under all \evalnum workloads as configuration input to \sysname, \sysabbr-$2$d and \sysabbr-dt for request dispatching. \figref{fig:ablation_1_combined}(a) shows the average normalized \codesub{L}{wc}. Given the \textit{same} configuration, \sysabbr-$2$d's RR dispatch policy and \sysabbr-dt's throughput-based dispatch policy require an extra latency of $90.4\%$ and $42.8\%$ on average, indicating \sysname's TC dispatch's capability to minimize \codesub{L}{wc}.

Moreover, to verify TC dispatch's capability on increasing throughput, \figref{fig:ablation_1_combined}(b) shows the average normalized throughput of three given modules, where the throughput for \sysabbr-$2$d and \sysabbr-dt is $40.7-45.6\%$ and $22.7-28.0\%$ lower than \sysname's. The alternative systems have to choose batch sizes with smaller throughput due to their inefficient request dispatch policy. Meanwhile, with the minimum \codesub{L}{wc}, \sysname can choose batch sizes with larger throughput by leveraging the saved latency budget, leading to a smaller serving cost.

\parab{The importance of multiple configurations.} \sysname supports multiple configurations for each module to maximize the throughput. To quantify its benefit, we compare \sysname against: (1) \sysabbr-$1$c, which assigns each module one configuration~\cite{crankshaw2017clipper, crankshaw2020inferline}, and (2) \sysabbr-$2$c, which assigns each module with two-tuple configurations~\cite{shen2019nexus, hu2021scrooge}.


\begin{figure}[t]
    \subfloat{\includegraphics[width=0.235\textwidth]{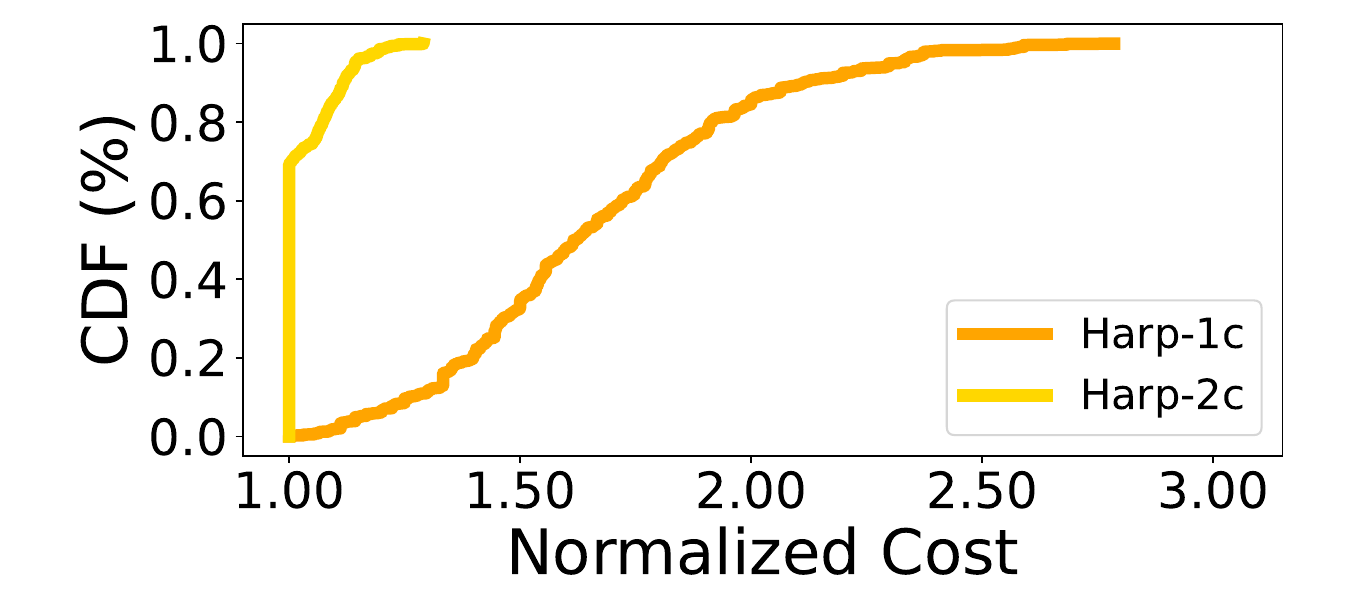}}
    \subfloat{\includegraphics[width=0.235\textwidth]{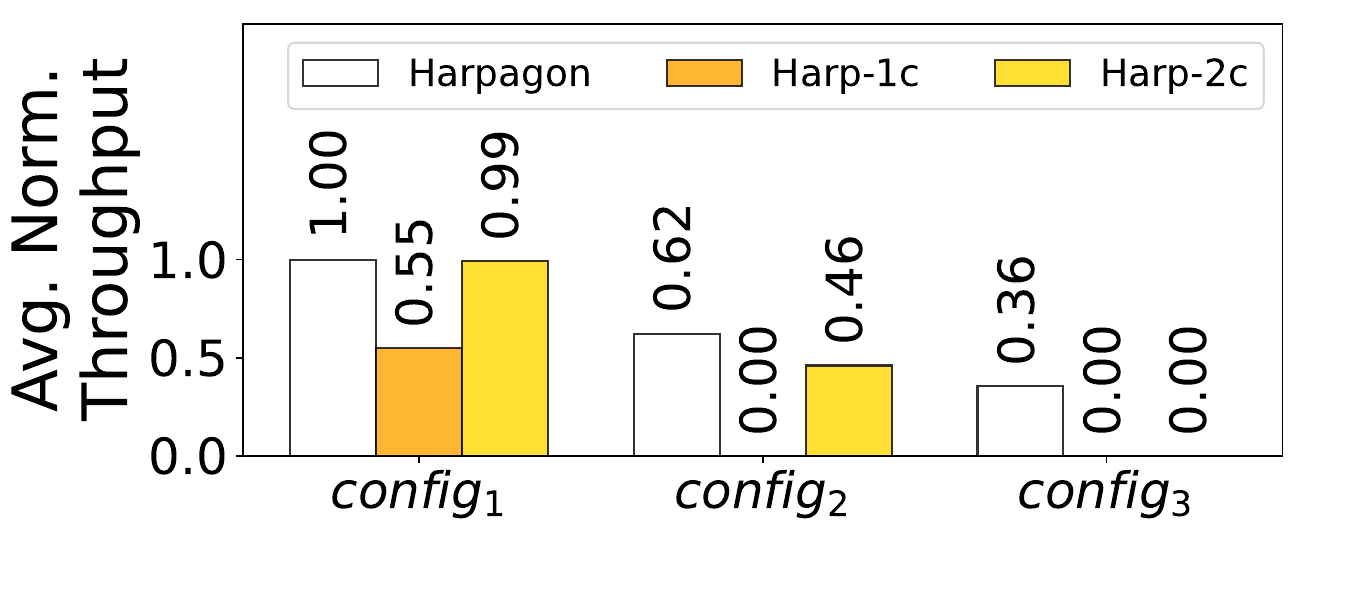}}
    \caption{(a) The CDF of normalized cost for \sysabbr-$1$c and \sysabbr-$2$c; (b) The average normalized throughput of configurations for \sysabbr-$1$c and \sysabbr-$2$c.}
    \label{fig:ablation_2_combined}\tight
\end{figure}

\figref{fig:ablation_all_a} includes the normalized cost for \sysabbr-$1$c and \sysabbr-$2$c, which require an average of $66.5\%$ and $3.0\%$ extra cost. \figref{fig:ablation_2_combined}(a) shows the CDF of the normalized cost. \sysabbr-$1$c has a larger serving cost for almost \textit{all} \evalnum workloads and a maximum of $178.6\%$ extra cost, while \sysabbr-$2$c has the same cost as \sysname for $67.6\%$ workloads and a maximum of $29.0\%$ extra cost. To understand why, \figref{fig:ablation_2_combined}(b) shows the average normalized throughput for a given module. The throughput for \sysabbr-$1$c's sole configuration is $45.0\%$ lower than \sysname's.  The throughput for \sysabbr-$2$c's first configuration is almost the same as \sysname's, but its second one is $26.1\%$ lower than \sysname's due to latency constraint. Among all workloads, \sysname assigns $32.4\%$ of them with more than two configurations to maximize the throughput, leading to a lower serving cost.




\parab{The importance of batching and heterogeneity.} \sysname selects proper batch size and computation hardware to minimize the serving cost. To quantify its importance, we compare \sysname against: (1) \sysabbr-nb, which disables batching, (2) \sysabbr-nhc, which always selects the cheapest hardware, and (3) \sysabbr-nhe, which selects the most expensive one.

\figref{fig:ablation_all_a} includes the normalized cost for the three alternatives. \sysabbr-nb has the highest serving cost with an average extra cost of $89.6\%$, suggesting the importance of batching on achieving the cost minimum goal. \sysabbr-nhc and \sysabbr-nhe require an average of $23.2\%$ and $14.0\%$ extra serving cost, which also indicates the necessity of selecting the proper hardware to reduce the cost. Both batching and heterogeneity reduce the serving cost by increasing the module throughput under the latency constraint. \figref{fig:ablation_3a} shows the average throughput for a given module under \evalnum workloads, where the throughput for the alternatives is $68.0\%$, $31.2\%$ and $7.2\%$ lower than \sysname's. Interestingly, \sysabbr-nhe's throughput for the given module is \textit{larger} than \sysname's among $4.9\%$ workloads. This is because after disabling heterogeneity, \sysabbr-nhe splits the latency differently, leaving the given module a larger latency budget. For all $4.9\%$ workloads, the average throughput of \sysabbr-nhe for the rest modules is much lower than \sysname's, leading to a higher total cost.


\begin{figure}[t]
\begin{minipage}[t]{0.48\linewidth}
    \includegraphics[width=\linewidth]{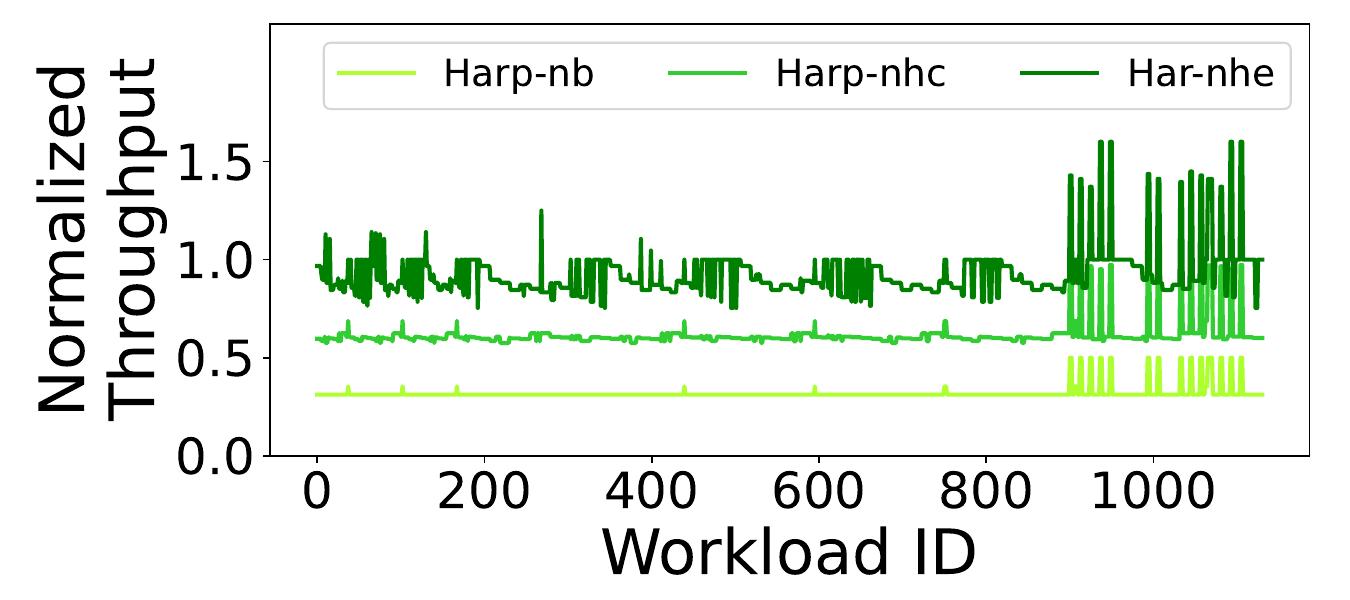}
    \caption{The normalized throughput for \sysabbr-nb, \sysabbr-nhc and \sysabbr-nhe.}
    \label{fig:ablation_3a}
\end{minipage}%
    \hfill%
\begin{minipage}[t]{0.48\linewidth}
    \includegraphics[width=\linewidth]{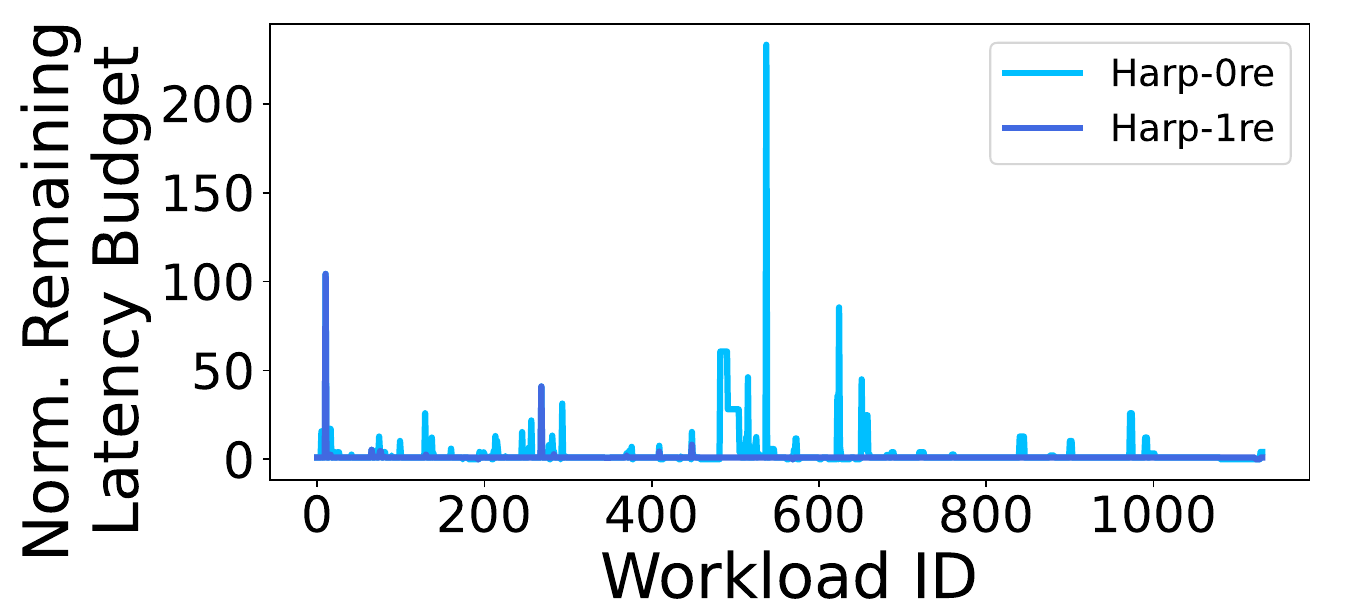}
    \caption{The norm. remaining latency budget for \sysabbr-$0$re and \sysabbr-$1$re.}
    \label{fig:ablation_5a}
\end{minipage} \tight
\end{figure}



\parab{The importance of optimizing residual.} As discussed in \secref{sec:scheduling}, \sysname adds dummy requests and reassigns the remaining latency budget to optimize the residual workload.


To quantify the benefit of adding dummy requests, we compare \sysname against \sysabbr-nd, which does not add dummy requests. \figref{fig:ablation_all_a} includes the normalized cost of \sysabbr-nd, which requires $0.8\%$ extra cost on average. For all \evalnum workloads, \sysname adds dummy request for $15.8\%$ of them, among which dummy requests increase the throughput of residual workload by $80.8\%$ with $10.9\%$ extra computation resources, leading to a $4.8\%$ cost reduction.


To quantify the benefit of reassigning the remaining latency budget, we compare \sysname against: (1) \sysabbr-$0$re, which does not reassign the latency budget, and (2) \sysabbr-$1$re, which reassigns the latency budget to one module greedily. \figref{fig:ablation_all_a} includes the normalized cost for the two alternatives, which require an average of $1.0\%$ and $0.6\%$ and a maximum of $14.6\%$ and $12.5\%$ extra cost. \figref{fig:ablation_5a} shows the remaining latency budget, where the one for \sysabbr-$0$re and \sysabbr-$1$re is $2.93$ and $1.14$ times of \sysname's with a maximum of $233.3$ and $104.3$ times respectively. Among all workloads, \sysname reassigns the remaining latency budget for at least once for $23.0\%$ workloads to further reduce the cost.



\parab{The importance of latency-cost efficiency}. \sysname leverages latency-cost efficiency to split latency across modules. To quantify its benefit, we compare \sysname against \sysabbr-tb, which uses module throughput to split latency~\cite{hu2021scrooge, crankshaw2020inferline}.



\begin{figure}[t]
\begin{minipage}[t]{0.48\linewidth}
    \includegraphics[width=\linewidth]{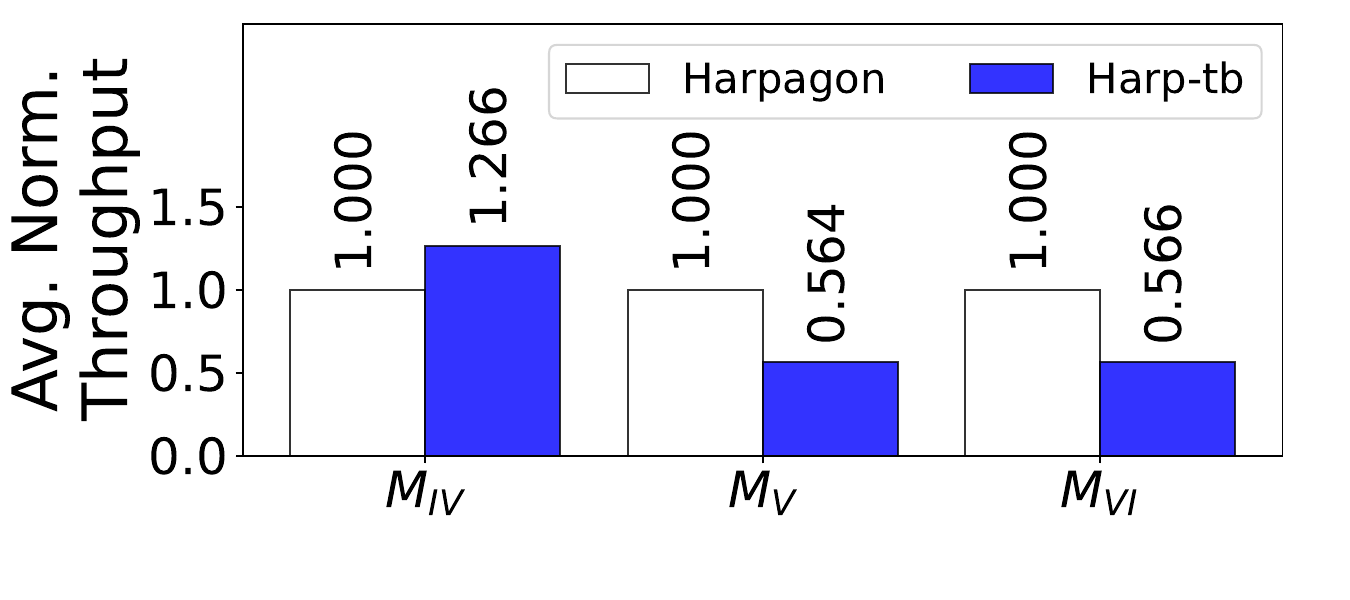}
    \caption{The average normalized throughput for a three-module app.}
    \label{fig:ablation_6b}
\end{minipage}%
    \hfill%
\begin{minipage}[t]{0.48\linewidth}
    \includegraphics[width=\linewidth]{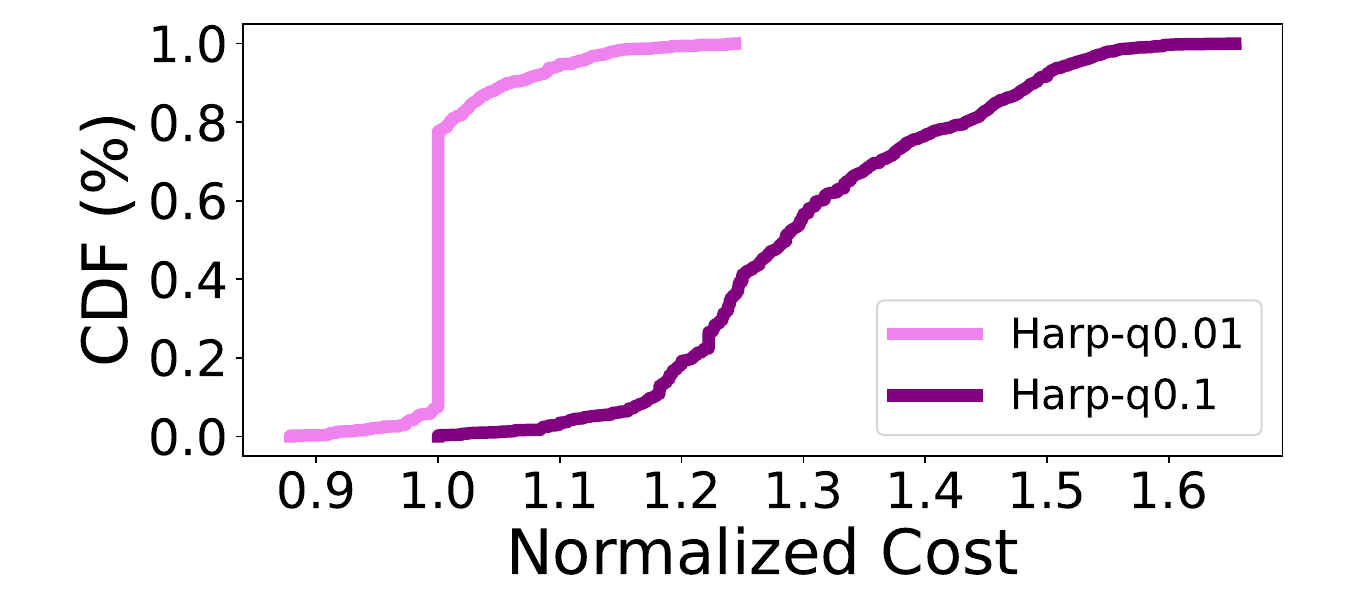}
    \caption{The CDF of normalized cost for \sysabbr-q$0.01$ and \sysabbr-q$0.1$.}
    \label{fig:ablation_7a}
\end{minipage} \tight
\end{figure}

\figref{fig:ablation_all_a} includes the normalized cost for \sysabbr-tb, which requires an average of $35.3\%$ extra cost. We measure the required iterations to derive the splitting results for \sysname and \sysabbr-tb. For all workloads, \sysname uses $10.9$ iterations on average, while \sysabbr-tb only uses $3.2$ iterations. As discussed in \secref{sec:splitting}, \sysname uses latency-cost efficiency to \textit{gradually} allocate the latency budget across modules, which avoids achieving sub-optimal results as \sysabbr-tb. 

Besides, we find that the advantage of \sysname over \sysabbr-tb is larger for applications with more modules. \figref{fig:ablation_6b} shows the average normalized throughput for a three-module application. \sysabbr-tb assigns most of its latency budget to module \codesub{M}{IV}, so as to maximize the corresponding throughput. However, the remaining latency budget for the rest modules is limited. Namely, the alternative tends to allocate more latency budget to modules with higher throughput, which ignores the marginal cost reduction of latency budget. Meanwhile, \sysname considers the amount of cost that per unit of latency budget can save, which enables \sysname to minimize the cost, especially for applications with multiple modules.

\parab{The importance of not quantized.} To derive per-module latency budget, an alternative is to quantize latency SLO into discrete intervals~\cite{shen2019nexus}. To show the benefit of \sysname, we compare it against: (1) \sysabbr-q$0.01$, which quantizes latency SLO into discrete interval in step of $0.01$ sec, and (2) \sysabbr-q$0.1$, which quantizes in step of $0.1$ sec.

\figref{fig:ablation_all_a} includes the normalized cost for the two alternatives, while \figref{fig:ablation_7a} shows the corresponding CDF. \sysabbr-q$0.1$ requires an average of $30.6\%$ and a maximum of $65.5\%$ extra serving cost. This is because the $0.1$-sec discrete interval is too coarse to split the latency SLO properly. Meanwhile, the serving cost of \sysabbr-q$0.01$ is close to \sysname's with an average of $1.2\%$ and a maximum of $24.4\%$ extra cost. Interestingly, for $7.3\%$ workloads, the serving cost of \sysabbr-q$0.01$ is smaller than \sysname's. This is because the quantized method is another kind of brute force search. However, quantized-based solution is too slow. The average runtime of \sysabbr-q$0.01$ is $2839$ ms, while \sysname's is only $5$ ms. Namely, though \sysabbr-q$0.01$ has lower cost for $7.3\%$ workloads, it has higher cost for $23.3\%$ workloads and a $567$ times runtime complexity.

\parab{The importance of splitting optimization.} \sysname's latency-cost efficiency based method can \textit{not} guarantee optimality, since latency splitting is NP-Complete. \sysname supports two splitting optimizations. To quantify the benefit of node merging and cost direct, we compare \sysname against: (1) \sysabbr-nnm, which disables node merging, and (2) \sysabbr-ncd, which disables cost direct. \figref{fig:ablation_all_a} includes the normalized cost for \sysabbr-nnm and \sysabbr-ncd, which requires an average of $0.2\%$ and $0.3\%$ extra serving cost. The alternatives benefit $3.98\%$ and $5.07\%$ workloads, among which they reduce $4.1\%$ and $5.2\%$ serving costs on average.

\section{Related Work}
\label{sec:related}
\parab{DNN Inference Systems.} Multiple DNN inference systems have been proposed recently, but \sysname differs from them with its unique designs. Nexus~\cite{shen2019nexus} is designed to schedule DNN models under latency constraints. However, it uses round-robin dispatch policy, which can not achieve the cost-minimum goal. Scrooge~\cite{hu2021scrooge} supports multi-dimensional scheduling, however, it only considers a maximum of two configurations and uses throughput-based heuristic to split the latency, leading to sub-optimal solution. InferLine~\cite{crankshaw2020inferline} and Clipper~\cite{crankshaw2017clipper} only consider one configuration per module. INFaaS~\cite{romero2021infaas} chooses appropriate DNN model among model variants by using a state machine to track variant performance, but it has limited support to split the latency for multi-DNN applications. Edge systems~\cite{hu2021rim, zeng2024efficient, she2023demand, wang2024Gecko, wang2024minimizing} are designed for small clusters and cannot be applied for cloud directly.




\parab{DNN Training Systems.} Gandiva~\cite{xiao2018gandiva} and AntMan~\cite{xiao2020antman} use spatial multiplexing and heterogeneous hardware to accelerate training process. Tiresias~\cite{gu2019tiresias} and Themis~\cite{mahajan2020themis} schedule training jobs in the cluster to minimize the job completion time. The design of \sysname is inspired by various DNN training systems~\cite{lao2021atp, zhao2022multi}, however, their scheduling capability can not satisfy inference job's millisecond latency objectives.

\section{Conclusion}
\label{sec:conclusion}
In this paper, we presented a cost-minimum DNN inference system called \sysname, which serves DNN workload efficiently to minimize the serving cost while satisfying the latency objectives. \sysname dispatches batched requests across machines to minimize the worst case latency of candidate configurations, which enables it to choose configurations with larger throughput. Besides, \sysname schedules modules with multiple configurations and leverages residual optimizer to maximize the throughput for both majority and residual workload. Moreover, for multi-DNN applications, \sysname splits the end-to-end latency into per-module latency budget using latency-cost efficiency and supports splitting optimizer to maximize the latency efficiency. Evaluation shows that \sysname outperforms the state of the art by $1.49$ to $2.37$ times in serving cost on average while satisfying the latency objective. \sysname derives the lower bound cost for more than $91\%$ workloads with millisecond-level runtime.

\bibliographystyle{unsrt}
\bibliography{references.bib}


\end{document}